\documentclass[11pt,a4paper]{article}
\usepackage{amsmath,amsfonts,amssymb,amsthm,latexsym}
\usepackage{graphics,epsfig,color,enumerate}
\usepackage{setspace,mathtools,enumitem,dsfont,verbatim}
\usepackage{caption}
\usepackage{subcaption}
\usepackage[english]{babel}
\makeatletter
\def\@captype{figure}
\makeatother

\newtheorem{theorem}{Theorem}[section]
\newtheorem{formula}{Formula}[section]
\newtheorem{lemma}[theorem]{Lemma}

\newtheorem{remark}[theorem]{Remark}
\newtheorem{definition}[theorem]{Definition}

\numberwithin{equation}{section}

\newcommand{\N}{\mathbb{N}}

\newcommand{\cG}{\mathcal{G}}

\newcommand{\vC}{\vec{C}}
\newcommand{\vt}{\vec{\theta}}

\newcommand{\Pmic}{P_{\mathrm{mic}}}
\newcommand{\Pcan}{P_{\mathrm{can}}}

\newcommand{\be}{\begin{equation}}
\newcommand{\ee}{\end{equation}}

\title{
Is breaking of ensemble equivalence monotone\\ in the number of constraints?}

\author{

Andrea Roccaverde
\footnotemark[1]\,\,\,\,\footnotemark[2]
}

\footnotetext[1]{
Lorentz Institute for Theoretical Physics, Leiden University, P.O.\  Box 9504, 
2300 RA Leiden, The Netherlands
}

\footnotetext[2]{
Mathematical Institute, Leiden University, P.O.\ Box 9512,
2300 RA Leiden, The Netherlands
}

\date{\today}

\begin{document}

\maketitle 

\begin{abstract}
Breaking of ensemble equivalence between the microcanonical ensemble and the canonical ensemble may occur for random graphs whose size tends to infinity, and is signaled by a non-zero specific relative entropy between the two ensembles. In \cite{GHR17} and \cite{GHR18} it was shown that breaking occurs when the constraint is put on the degree sequence (configuration model). It is not known what is the effect on the relative entropy when the number of constraints is reduced, i.e., when only part of the nodes are constrained in their degree (and the remaining nodes are left unconstrained). Intuitively, the relative entropy is expected to decrease. However, this is not a trivial issue because when constraints are removed both the microcanonical ensemble and the canonical ensemble change. In this paper a formula for the relative entropy valid for generic discrete random structures, recently formulated by Squartini and Garlaschelli, is used to prove that the relative entropy is monotone in the number of constraints when the constraint is on the degrees of the nodes. It is further shown that the expression for the relative entropy corresponds, in the dense regime, to the degrees in the microcanonical ensemble being asymptotically \emph{multivariate Dirac} and in the canonical ensemble being asymptotically \emph{Gaussian}.



\medskip\noindent
{\it Acknowledgment.} 
I thank Diego Garlaschelli and Frank den Hollander for the many discussions we have shared on breaking of ensemble equivalence.
\\
The research in this paper is supported by EU-project 317532-MULTIPLEX and by NWO 
Gravitation Grant 024.002.003--NETWORKS.

\end{abstract}

\newpage


\section{Introduction and main results}
\label{S1}

\subsection{Background}
\label{S1.1}

For most real-world networks, a detailed knowledge of the architecture of the network is not available and one must work with a probabilistic description, where the network 
is assumed to be a random sample drawn from a set of allowed configurations that are consistent with a set of known \emph{topological constraints}~\cite{SMG15}. Statistical physics deals with the definition of the appropriate probability distribution over the set 
of configurations and with the calculation of the resulting properties of the system. Two key choices of probability distribution are: 
\begin{itemize}
\item[(1)] 
the \emph{microcanonical ensemble}, where the constraints are \emph{hard} (i.e., are 
satisfied by each individual configuration); 
\item[(2)] 
the \emph{canonical ensemble}, where the constraints are \emph{soft} (i.e., hold as 
ensemble averages, while individual configurations may violate the constraints).
\end{itemize} 
(In both ensembles, the entropy is \emph{maximal} subject to the given constraints.)
   
Breaking of ensemble equivalence means that different choices of the ensemble lead to asymptotically different behaviors. Consequently, while for applications based on ensemble-equivalent models the choice of ensemble can be based on mathematical convenience, for those based on ensemble-nonequivalent models the choice should be determined by the system one wants to apply to, i.e., dictated by a theoretical criterion that indicates \emph{a priori} which ensemble is the appropriate one to be used. It is known that ensemble equivalence may be broken, signaled by a non-zero specific relative entropy between the two ensembles. It is expected that when the number of constraints grows extensively in the number of nodes, then typically there is breaking of ensemble equivalence. This has been shown to be the case when the setting is simple or bipartite graphs and the constraint is on the number of links (1 constraint and ensemble equivalence) or on the full degree sequence ($n$ constraints and non-equivalence) \cite{GHR17}. Later, in \cite{GHR18} and \cite{dHMRS17}, also the dense regime was investigated and it was shown that the relative entropy between the two ensembles grows even faster. 
In general, the constraint is a multidimensional vector and its components represent the single quantities that are constrained. From now on, with the word `constraint' we mean the `vector constraint' and with the plural `constraints' we mean the `components' of the vector. This means when we talk about the number of constraints we actually mean the dimension of the vector constraint. In some cases this number can be very large, for example, when the constraint is on the degree sequence (a large number of nodes which need all to have the right degree).

Once the constraint becomes a function of the number $n$ of nodes (for example, the degree sequence), we can ask an interesting question: How is the relative entropy affected when the number of constraints is reduced, possibly in a way that depends on $n$? Intuitively, the relative entropy should decrease, but this is not a trivial issue because both the microcanonical and the canonical ensemble change when the constraints are changed. Of particular interest for the present paper is the main result of \cite{GHR18}. There it was proven that, when a $\delta$-tame degree sequence is put as a constraint on the set of simple graphs, than the relative entropy between the two ensembles grows as $n\log n$. We consider random graphs with a prescribed partial degree sequence (reduced constraint). The breaking of ensemble equivalence is studied by analyzing how the relative entropy changes as a function of the number of constraints, in particular, it is shown that the relative entropy is a monotone function of the number of constraints. More precisely, when only $m$ nodes are constrained and the remaining $n-m$ nodes are left unconstrained, the relative entropy is shown to grow like $m\log n$. Our analysis is based on a recent formula put forward by Squartini and Garlaschelli~\cite{GS}. This formula predicts that the relative entropy is determined by the covariance matrix of the constraints under the canonical ensemble, in the regime where the graph is dense. Our result implies that ensemble equivalence breaks down whenever the regime is $\delta$-tame, irrespective of the number of degrees $m$ that are constrained, provided $m$ is not of order $n$. 

\subsection*{Outline}
Our paper is organized as follows. In Section \ref{S1} the background, the model and the main theorem are discussed. In Section \ref{S2} the main theorem is proved, together with a few basic lemmas that are needed along the way. Appendix \ref{appA} derives an expression for the canonical ensemble when a partial degree sequence is put as constraint. Appendix \ref{appB} discusses the $\delta$-tame condition for a partial degree sequence. 

The remainder of Section \ref{S1} is organized as follows. In Section \ref{S1.1} we discussed the background of the problem. In Section~\ref{S1.2} we define the two ensembles and their relative entropy. In Section \ref{S1.3} we describe the model when the constraint is put on the full degree sequence, in Section~\ref{S1.4} when the constraint is put on the partial degree sequence. Here we also define the $\delta$-tame regime when the constraint is on the partial degree sequence. In Section~\ref{S1.5} we state a formula for the relative entropy presented in \cite{GS} and state the main theorem. In section \ref{S1.6} we interpret the main theorem by stating how the degrees are distributed in the two ensembles. 

\subsection{Microcanonical ensemble, canonical ensemble, relative entropy}
\label{S1.2}

This section defines the two ensembles and their relative entropy when the configuration space is the set of simple graphs with $n$ nodes and the constraint is general. The content of this section is borrowed from \cite[Section 1.2]{GHR18}. For $n \in \N$, let $\cG_n$ denote the set of all simple undirected graphs with $n$ nodes.
Any graph $G\in\cG_n$ can be represented as an $n \times n$ matrix with elements 
\be
g_{ij}(G) =
\begin{cases}
1,\qquad \mbox{if there is a link between node\ } i \mbox{\ and node\ } j,\\ 
0,\qquad \mbox{otherwise.}
\end{cases}
\ee
Let $\vC$ denote a vector-valued function on $\cG_n$. Given a specific value $\vC^*$, 
which we assume to be \emph{graphical}, i.e., realizable by at least one graph in $\cG_n$, 
the \emph{microcanonical probability distribution} on $\cG_n$ with \emph{hard constraint} 
$\vC^*$ is defined as
\begin{equation}
\Pmic(G) =
\left\{
\begin{array}{ll} 
\Omega_{\vC^*}^{-1}, \quad & \text{if } \vC(G) = \vC^*, \\ 
0, & \text{else},
\end{array}
\right.
\label{eq:PM}
\end{equation}
where 
\begin{equation}
\Omega_{\vC^*} = | \{G \in \cG_n\colon\, \vC(G) = \vC^* \} |
\end{equation}
is the number of graphs that realise $\vC^*$. The \emph{canonical probability distribution} 
$\Pcan(G)$ on $\cG_n$ is defined as the solution of the maximisation of the 
\emph{entropy} 
\begin{equation}
S_n(\Pcan) = - \sum_{G \in \cG_n} \Pcan(G) \ln \Pcan(G)
\end{equation}
subject to the normalisation condition $\sum_{G \in \cG_n} \Pcan(G) = 1$ and to the 
\emph{soft constraint} $\langle \vC \rangle  = \vC^*$, where $\langle \cdot \rangle$ 
denotes the average w.r.t.\ $\Pcan$. This gives
\begin{equation}
\Pcan(G) = \frac{\exp[-H(G,\vt^*)]}{Z(\vt^*)},
\label{eq:PC}
\end{equation}
where 
\begin{equation}
H(G, \vt\,) = \vt \cdot \vC(G)
\label{eq:H}
\end{equation}
is the \emph{Hamiltonian} and
\be
Z(\vt\,) = \sum_{G \in \cG_n} \exp[-H(G, \vt\,)]
\ee
is the \emph{partition function}. In \eqref{eq:PC} the parameter $\vt$ must be set equal to 
the particular value $\vt^*$ that realises $\langle \vC \rangle  = \vC^*$. This value is unique
and maximises the likelihood of the model given the data (see \cite{GL08}).

The \emph{relative entropy} of $\Pmic$ w.r.t.\ $\Pcan$ is~\cite{T15}
\begin{equation}
S_n(\Pmic \mid \Pcan) 
= \sum_{G \in \cG_n} \Pmic(G) \log \frac{\Pmic(G)}{\Pcan(G)},
\label{eq:KL1}
\end{equation}
and the \emph{relative entropy $\alpha_n$-density} is~\cite{GS} 
\begin{equation}
s_{\alpha_n} = {\alpha_n}^{-1}\,S_n(\Pmic \mid \Pcan),
\label{eq:sn}
\end{equation}
where $\alpha_n$ is a \emph{scale parameter}. The limit of the relative entropy 
$\alpha_n$-density is defined as
\begin{equation}
s_{\alpha_\infty}\equiv\lim_{n \to \infty}s_{\alpha_n} 
= \lim_{n \to \infty} {\alpha_n}^{-1}\,S_n(\Pmic \mid \Pcan) \in [0,\infty],
\label{eq:criterion1}
\end{equation}
We say that the microcanonical and the canonical ensemble are equivalent on scale $\alpha_n$ 
if and only if \footnote{As shown in \cite{T15} within the context of interacting particle systems, 
relative entropy is the most sensitive tool to monitor breaking of ensemble equivalence (referred 
to as breaking \emph{in the measure sense}). Other tools are interesting as well, depending on 
the `observable' of interest~\cite{T18}.}    
\be
s_{\alpha_\infty} = 0.
\label{salphao}
\ee
We recall that here with `constraint' we mean the `vector constraint' and with `constraints' we mean the `components' of the vector.
The choice of $\alpha_n$ is flexible. The natural choice is the one for which $s_{\alpha_\infty} 
\in (0,\infty)$, and depends on the constraint at hand as well as its value. For instance, if the 
constraint is on the \emph{degree sequence}, then in the sparse regime the natural scale turns 
out to be $\alpha_n=n$ \cite{SdMdHG15}, \cite{GHR17} (in which case $s_{\alpha_\infty}$ is 
the specific relative entropy `per vertex'), while in the dense regime it turns out to be $\alpha_n 
= n\log n$ \cite{GHR18}. On the other hand, if the constraint is on the \emph{total numbers 
of edges and triangles}, with values different from what is typical for the Erd\H{o}s-Renyi random 
graph in the dense regime, then the natural scale turns out to be $\alpha_n=n^2$ \cite{dHMRS17} 
(in which case $s_{\alpha_\infty}$ is the specific relative entropy `per edge'). Such a severe 
breaking of ensemble equivalence comes from `frustration' in the constraints.  

Before considering specific cases, we recall an important observation made in \cite{SdMdHG15}. 
The definition of $H(G,\vt\,)$ ensures that, for any $G_1,G_2\in\cG_n$, $\Pcan(G_1)=\Pcan(G_2)$ 
whenever $\vC(G_1)=\vC(G_2)$ (i.e., the canonical probability is the same for all graphs having 
the same value of the constraint). We may therefore rewrite \eqref{eq:KL1} as
\begin{equation}
S_n(\Pmic \mid \Pcan) = \log \frac{\Pmic(G^*)}{\Pcan(G^*)},
\label{eq:KL2}
\end{equation}
where $G^*$ is \emph{any} graph in $\cG_n$ such that $\vC(G^*) =\vC^*$ (recall that we 
have assumed that $\vC^*$ is realisable by at least one graph in $\cG_n$). The definition 
in~\eqref{eq:criterion1} then becomes  
\begin{equation}
s_{\alpha_\infty}=\lim_{n \to \infty} {\alpha_n}^{-1}\, \big[\log{\Pmic(G^*)} - \log{\Pcan(G^*)} \big],
\label{eq:criterion3}
\end{equation}
which shows that breaking of ensemble equivalence coincides with $\Pmic(G^*)$ and $\Pcan(G^*)$ 
having different large deviation behaviour on scale $\alpha_n$. Note that \eqref{eq:criterion3} 
involves the microcanonical and canonical probabilities of a \emph{single} configuration $G^*$ 
realising the hard constraint. Apart from its theoretical importance, this fact greatly simplifies computations. 
To analyse breaking of ensemble equivalence, ideally we would like to be able to identify an 
underlying \emph{large deviation principle} on a natural scale $\alpha_n$. This is generally 
difficult, and so far has only been achieved in the dense regime with the help of \emph{graphons} 
(see \cite{dHMRS17} and references therein). In the present paper we will approach 
the problem from a different angle, namely, by looking at the \emph{covariance matrix of the 
constraints} in the canonical ensemble, as proposed in \cite{GS}.      

Note that all the quantities introduced above in principle depend on $n$. However, except for 
the symbols $\cG_n$ and $S_n(\Pmic \mid \Pcan)$, we suppress the $n$-dependence from 
the notation.
 
\subsection{Constraint on the full degree sequence}
\label{S1.3}
The model of this section comes from \cite{GHR17} and \cite{GHR18}. The full degree sequence of a graph $G\in \cG_n$ is defined as the vector $\vec{k}(G) = (k_i(G))_{i=1}^n$ with $k_i(G)=\sum_{j \neq i}g_{ij}(G)$. The degree sequence is set to a specific value $\vec{k}^*$, which we  assume to be graphical, i.e., there is at least one graph with degree sequence $\vec{k}^*$.
The constraint is therefore
\be 
\vC^* = \vec{k}^*= (k_i^*)_{i=1}^n \in \{1,2,\dots ,n-2\}^n.
\ee 
This constraint was studied in various regimes: in \cite{GHR17} in the sparse regime, and in \cite{GHR18} in the ultra-dense and the $\delta$-tame regime. The microcanonical ensemble, when the constraint is put on the degree sequence, is known as the \emph{configuration model} and has been studied in detail (see \cite{SMG15,SdMdHG15,vdH17}). In the sparse (and in the ultra-dense) regime, the microcanonical ensemble cannot be computed exactly, but there are good approximations with an error that is vanishing when the relative entropy is computed in the limit as $n\to\infty$ \cite{GHR17}, \cite{BH}. In the $\delta$-tame regime, this approximation does not hold, but the relative entropy can still be investigated with other tools \cite{GHR18}. The canonical ensemble can be computed in every regime and takes the form
\begin{equation}
\label{canonicalO}
\Pcan(G) = \prod_{1 \leq i<j \leq n}\left( p_{ij}^* \right)^{g_{ij}(G)} \left( 1- p_{ij}^* \right)^{1-g_{ij}(G)},
\end{equation}
with 
\be
\label{canonical1O} 
p_{ij}^* = \frac{e^{-\theta_i^*-\theta_j^*}}{1 + e^{-\theta_i^*-\theta_j^*}},
\ee 
and with the vector of Lagrange multipliers $\vec{\theta}^*=(\theta_i^*)_{i=1}^n$ tuned such that 
\begin{equation}
\label{canonical2O}
\langle k_i \rangle = \sum_{1\le j\le n\atop j \neq i}p_{ij}^* = k_i^*, \qquad \forall\ 1\le i\le n.
\end{equation}
The results in \cite{GHR17} show that there is breaking of ensemble equivalence with $\alpha_n=n$ when the regime is sparse and ultra-dense. The results in \cite{GHR18} show that the relative entropy grows like $\alpha_n = \tfrac{1}{2}n\log n$. The purpose of the paper is to investigate what happens when part of the $n$ constraints degrees are removed and how the relative entropy is affected by this. In the next section the partial constraint is presented and the main theorem is stated.   

\subsection{Constraint on the partial degree sequence}
\label{S1.4}
In this section we look at a different model. The constraint is put on the partial degree sequence instead of on the full degree sequence, more precisely, only the first $m<n$ nodes are constrained while the remaining nodes are left unconstrained. The partial degree sequence  of a graph $G\in\cG_n$ is defined as the vector $\vec{k}(G) = (k_i(G))_{i=1}^m$ where $k_i(G)=\sum_{1\le j\neq j\le n}g_{ij}(G)$. The constraint is set to be a \emph{specific $m$-dimensional} vector $\vec{k}^*$, which we assume to be \emph{graphical}, i.e., there exist at 
least one graph $G^*\in\cG_n$ with partial degree sequence $\vec{k}^*$. The constraint is therefore
\be
\label{degcon}
\vC^* = \vec{k}^*= (k_i^*)_{i=1}^m \in \{1,2,\dots ,n-2\}^m,
\ee 
As mentioned above, the microcanonical ensemble can be computed approximately when the constraint is put on the full degree sequence. However, when the constraint is put on the partial degree sequence, no good approximation is available. The situation is different for the canonical ensemble, which can still be computed. Appendix \ref{appA} is dedicated to the study of the canonical ensemble when a partial degree sequence is put as a constraint. This leads to 
\begin{equation}
\label{canonical}
\Pcan(G) = 2^{-{{n-m}\choose 2}}\prod_{1 \leq i<j \leq m}\left( p_{ij}^* \right)^{g_{ij}(G)} \left( 1- p_{ij}^* \right)^{1-g_{ij}(G)}\prod_{i=1}^m\left( p_{i}^* \right)^{s_{i}(G)} \left( 1- p_{i}^* \right)^{n-m-s_{i}(G)}
\end{equation}
with 
\be
\label{canonical1} 
p_{ij}^* = \frac{e^{-\theta_i^*-\theta_j^*}}{1 + e^{-\theta_i^*-\theta_j^*}},\qquad p_{i}^* = \frac{e^{-\theta_i^*}}{1 + e^{-\theta_i^*}},\qquad s_i(G)=\sum_{j=m+1}^n g_{ij}(G),
\ee 
and with the vector of Lagrange multipliers $\vec{\theta}^*=(\theta_i^*)_{i=1}^m$ tuned such that 
\begin{equation}
\label{canonical2}
\langle k_i \rangle = \sum_{1\le j\le m\atop j \neq i}p_{ij}^* + (n-m)p_{i}^* = k_i^*, \qquad 1\le i\le m.
\end{equation}

The canonical ensemble has an interesting dual structure, consisting of the product of two canonical probabilities, which we call unipartite probability and bipartite probability, and an overall factor $2^{-{{n-m}\choose 2}}$. The unipartite probability, 
$$\prod_{1 \leq i<j \leq m}\left( p_{ij}^* \right)^{g_{ij}(G)} \left( 1- p_{ij}^* \right)^{1-g_{ij}(G)},$$
is precisely the canonical ensemble obtained when the constraint is put on the full degree sequence $\vec{u}^*= (u_i^*)_{i=1}^m$, with $u_i^* = \sum_{1\le j\le m\atop j \neq i}p_{ij}^*$, on the subset of graphs with $m$ nodes $\cG_m$. The bipartite probability,
$$\prod_{i=1}^m\left( p_{i}^* \right)^{s_{i}(G)} \left( 1- p_{i}^* \right)^{n-m-s_{i}(G)},$$
is precisely the canonical bipartite probability obtained when the constraint is put on the top layer of a bipartite graph. More precisely, the configuration space is the set of bipartite graphs $\cG_{m,n-m}$ with $m$ nodes in the top layer and $n-m$ nodes in the bottom layer. The constraint is put on the degree sequence in the top layer only and corresponds to the vector $\vec{b}^*= (b_i^*)_{i=1}^m$ with $b_i^* = (n-m)p_i^*$. Moreover, the average $i$-th degree $\langle k_i \rangle$ with respect to the canonical ensemble \eqref{canonical} equals $k_i^*$ and is given by the balance equation \eqref{canonical2}. This equation shows that the $i$-th unipartite constraint $u_i^*$ and the $i$-th bipartite constraint $b_i^*$ sum up to the $i$-th original constraint $k_i^*$. 

\begin{definition}
\label{delta}
A partial degree sequence $\vec{k}^*= (k_i^*)_{i=1}^m$, put as a constraint on the set of configurations $\cG_n$ with $m<n$, is said to be $\delta$-tame if and only if there exists 
a $\delta\in \left(0,\frac{1}{2}\right]$ such that
\begin{equation}
\label{deltatamedef}
\delta \leq p_{ij}^* \leq 1-\delta, \qquad 1\leq i \neq j \leq m,
\end{equation}
where $p_{ij}^*$ are the canonical probabilities in \eqref{canonical}--\eqref{canonical2}.
\end{definition}
It is easy to prove that, given a $\delta$-tame partial degree sequence $\vec{k}^*= (k_i^*)_{i=1}^m$, the bipartite probabilities $(p_{i}^*)_{i=1}^m$ are also $\delta$-tame, namely, satisfy
\be
\delta' \leq p_{i}^* \leq 1-\delta', \qquad \forall\ 1\leq i \leq m,
\ee
for some $\delta'\in \left(0,\frac{1}{2}\right]$. This is discussed in more detail in Appendix~\ref{appB}. Condition \eqref{deltatamedef} has a trivial implication for the degree sequence:
\begin{equation}
\label{deltatameondegrees1}
(m-1)\delta \leq u_i^* \leq (m-1)(1-\delta), \qquad 1\leq i \leq m,
\end{equation}
\begin{equation}
\label{deltatameondegrees2}
(n-m)\delta' \leq b_i^* \leq (n-m)(1-\delta'), \qquad 1\leq i \leq m.
\end{equation}
Since $\delta' = \frac{1}{1+(\frac{1-\delta}{\delta})^{3/2}} < \delta$ for all $\delta \in [0,1/2)$ and $u_i^* + b_i^* = k_i^*$, it follows that
\begin{equation}
\label{deltatameondegrees3}
(n-1)\delta'\leq k_i^* \leq (n-1)(1-\delta'), \qquad 1\leq i \leq m.
\end{equation}
This means that $\delta$-tame graphs are neither too thin (sparse regime) nor too dense (ultra-dense regime). 
It is natural to ask whether, conversely, condition \eqref{deltatameondegrees3}, or a similar condition involving only the original degrees $\vec{k}^*=(k_i^*)_{i=1}^m$, is sufficient to prove that the partial degree sequence is $\delta$-tame for some $\delta=\delta(\delta')$, in the sense of Definition \ref{delta}. Unfortunately, this question is not easy
to settle, but the following lemma provides a partial answer. 
\begin{lemma}
\label{reverse}
Suppose that $\vec{k}^*= (k_i^*)_{i=1}^m$ satisfies
\begin{equation}
\label{invdeltatame}
(n-1)\delta' + (n-m) \leq k_i^* \leq (n-1)(1-\delta'), \qquad 1\leq i \leq m,
\end{equation} 
for some $\delta' \in (\tfrac14,\tfrac12]$. Then there exist $\delta = \delta(\delta')>0$ and $n_0=n_0(\delta') 
\in \N$ such that $\vec{k}^*= (k_i^*)_{i=1}^m$ is a $\delta$-tame partial degree sequence, in the sense of Definition \ref{delta}, for all $n \geq n_0$.
\end{lemma}
\begin{proof}
Condition \eqref{invdeltatame}, with $u_i^* = k_i^*-b_i^*$ and $b_i^* \in [0,n-m]$, gives
\begin{equation}
\label{invdeltatameold}
(n-1)\delta' \leq u_i^* \leq (n-1)(1-\delta'), \qquad 1\leq i \leq m.
\end{equation}
The proof follows from \eqref{invdeltatameold} and \cite[Theorem 2.1]{BH}. In fact, applying that theorem with $\alpha=\delta'$, $\beta=1-\delta'$ and with $\delta'>\tfrac14$, we get
\begin{equation}
\label{deltatameunip}
\delta \leq p_{ij}^* \leq 1-\delta, \qquad 1\leq i \neq j \leq m.
\end{equation}
Moreover, \cite[Theorem 2.1]{BH} also gives information about the values of 
$\delta = \delta(\delta')$ and $n_0=n_0(\delta')$. 
\end{proof}


\subsection{Linking ensemble nonequivalence to the canonical covariances}
\label{S1.5}

In this section we describe an important formula, recently put forward in~\cite{GS}, for the scaling 
of the relative entropy under a general constraint. The analysis in \cite{GS} allows for the possibility 
that not all the constraints (i.e., not all the components of the vector $\vec{C}$) are linearly independent. 
For instance, $\vec{C}$ may contain redundant replicas of the same constraint(s), or linear combinations 
of them. Since in the present paper we only consider the case where $\vec{C}$ is the degree sequence, 
the different components of $\vec{C}$ (i.e., the different degrees) are linearly independent.

When a $K$-dimensional constraint $\vec{C}^* = (C^*_i)_{i=1}^K$ with independent components is 
imposed, then a key result in~\cite{GS} is the formula
\be
S_n(\Pmic \mid \Pcan) \sim \log\frac{\sqrt{\det(2\pi Q)}}{T}, \qquad n\to\infty,
\ee
where
\begin{equation}
\label{covQ}
Q=(q_{ij})_{1 \leq i,j \leq K}
\end{equation}
is the $K\times K$ covariance matrix of the constraints under the canonical ensemble, whose entries are 
defined as
\be
\qquad q_{ij} = \mathrm{Cov}_{\Pcan}(C_i,C_j)=\langle C_i\,C_j\rangle-\langle C_i\rangle \langle C_j\rangle,
\ee
and
\be
T=\prod_{i=1}^K\left[1+O\left(1/\lambda_i^{(K)}(Q)\right)\right], 
\label{eq:T}
\ee
with $\lambda_i^{(K)}(Q)>0$ the $i$-th eigenvalue of the $K\times K$ covariance matrix $Q$. This result can 
be formulated more rigorously as follows.
\begin{formula}[\cite{GS}]
\label{conj}
If all the constraints are linearly independent, then the limiting relative entropy ${\alpha_n}$-density 
equals
\be
s_{\alpha_\infty}=\lim_{n\to\infty}\frac{\log\sqrt{\det(2\pi Q)}}{\alpha_n}+\tau_{\alpha_\infty}
\label{eq:deltalimit}
\ee
with
\be
\tau_{\alpha_\infty}=-\lim_{n\to\infty}\frac{\log T}{\alpha_n}.
\label{eq:tau}
\ee
The latter is zero when 
\begin{equation}
\label{CondforFormula}
\lim_{n\to\infty} \frac{|I_{K_n,R}|}{\alpha_n}=0\quad \forall\,R<\infty,
\end{equation}
where $I_{K,R} = \lbrace i=1,\dots,K\colon\,\lambda_i^{(K)}(Q) \le R \rbrace$ with $\lambda_i^{(K)}(Q)$
the $i$-th eigenvalue of the $K$-dimensional covariance matrix $Q$ (the notation $K_n$ indicates that
$K$ may depend on $n$). Note that $0\le I_{K,R} \le K$. Consequently, \eqref{CondforFormula} is satisfied 
(and hence $\tau_{\alpha_\infty}=0$) when $\lim_{n\to\infty} K_n/\alpha_n=0$, i.e., when the number $K_n$ 
of constraints grows slower than $\alpha_n$.
\end{formula}

\begin{remark}[\cite{GS}]
{\rm Formula~{\rm \ref{conj}}, for which \cite{GS} offers compelling evidence but not a mathematical 
proof, can be rephrased by saying that the natural choice of $\alpha_n$
is
\be
\tilde{\alpha}_n=\log\sqrt{\det(2\pi Q)}.
\label{eq:alphatilde}
\ee
Indeed, if all the constraints are linearly independent and \eqref{CondforFormula} holds, 
then $\tau_{\tilde{\alpha}_n}=0$ and 
\begin{eqnarray}
\label{eq:deltalimit2}
&s_{\tilde{\alpha}_\infty}=1,\\
&S_n(\Pmic \mid \Pcan)=[1+o(1)]\,\tilde{\alpha}_n.
\label{conjecture2}
\end{eqnarray}
}
\end{remark}

Formula \ref{conj} has been verified in several examples, namely, all the models in \cite{GHR17} and \cite{GHR18}.

Next we present our main theorem, which considers the case where the constraint is on the partial degree sequence $\vC^*=\vec{k}^*= (k_i^*)_{i=1}^m$ in the \emph{$\delta$-tame regime} defined in Definition \ref{delta}.

\begin{theorem}
\label{MainTheorem}
Suppose that:
\begin{itemize}
\item
The constraint is put on the partial degree sequence $\vC^*= \vec{k}^*= (k_i^*)_{i=1}^m$ on the space of simple graphs $\cG_n$ with $0\le m \le n$.
\item
$\vC^*= \vec{k}^*= (k_i^*)_{i=1}^m$ is a $\delta$-tame partial degree sequence, namely, the canonical probabilities $(p_{ij}^*)_{1\le i \neq j\le m}$ satisfy
\be 
\delta \leq p_{ij}^* \leq 1-\delta, \quad 1 \leq i\neq j \leq m.
\ee
\item
Formula \ref{conj} is valid in the above framework.
\item
The scale parameter is $\alpha_n = \dfrac{m\log n}{2}$.
\item
$m=m(n)$ satisfies 
\be 
\label{Condmn}
\lim_{n\to \infty}\dfrac{n-m}{m}\log n = \infty.
\ee
\end{itemize}
Then there is breaking of ensemble equivalence, and  
\begin{equation}
\label{Result}
s_{{\alpha}_\infty}= \lim_{n \to \infty} s_{\alpha_n} = 1.
\end{equation}
\end{theorem}
Condition \eqref{Condmn} fails when $n-m=O(\tfrac{m}{\log n})$, i.e., when the number of unconstrained nodes is sufficiently small. We expect that \eqref{Result} continues to hold even in this case, but our proof breaks down.

\subsection{Discussion}
\label{S1.6}

Theorem \ref{MainTheorem} analyses the relative entropy at a macroscopic level, but says nothing about what happens at the microscopic level. More precisely, it does not identify how the relative entropy changes when a single constraint is removed, rather than a positive fraction of constraints. A microscopic analysis could reveal what is the effect when e.g. the longest degree is removed, or the smallest degree, or any other degree. The result in Theorem \ref{MainTheorem} is far from trivial. In fact, when the number of constraints is reduced, it can become either easier or more difficult to compute microcanonical and canonical ensembles. The case when the constraint is put on the degree sequence provides a clear example. If the constraint is put on the full degree sequence, then the microcanonical ensemble can be \emph{asymptotically computed} \cite{BH}. As soon as one or more degrees are removed (meaning that some nodes are left unconstrained), the structure of the problem changes completely. The symmetry of the constraints is broken by the removal, and this makes it more difficult to compute the number of graphs with a prescribed partial degree sequence. On the other hand, the canonical problem can still be solved and has an interesting structure (Appendix \ref{appA}). This makes it possible to use the formula proposed by Garlaschelli and Squartini \cite{GS}, which only makes use of the canonical ensemble to analyze the relative entropy between the two ensembles. Theorem \ref{MainTheorem} clearly exhibits the monotonicity property of the relative entropy in the case where the constraint is put on the degrees. Indeed, under the hypotheses written above, the relative entropy $S_n(\Pmic \mid \Pcan)$ grows like $m\log n$, where $m$ is the number of constrained nodes and $n$ is the total number of nodes.  
This shows that the relative entropy is monotone in the number of constraints on scale $n$. 

We next provide a heuristic explanation for Theorem \ref{MainTheorem} (in analogy with what was done in \cite{GHR17} and \cite{GHR18}). 
\paragraph{Heuristic explanation of Theorem \ref{MainTheorem}.} 
Using \eqref{eq:KL2}, we can write the relative entropy between the ensembles as
\be
\label{interp}
S_n(\Pmic \mid \Pcan) = \log \frac{\Pmic(G^*)}{\Pcan(G^*)} 
= -\log [\Omega^n_{\vec{k}^*}\Pcan(G^*)]= -\log Q^n[\vec{k^*}](\vec{k^*}),
\ee
where $\Omega^n_{\vec{k}^*}$ is the number of graphs with $n$ nodes and partial degree sequence $\vec{k}^*= (k_i^*)_{i=1}^m$,
\begin{equation}
Q^n[\vec{k^*}](\vec{k}\,) = \Omega_{\vec{k}}^n\,\Pcan\big(G^{\vec{k}}\big)
\label{eq:QOmega}
\end{equation}
is the probability that the partial degree sequence is equal to $\vec{k}$ under the canonical ensemble 
with constraint $\vec{k}^*$, $G^{\vec{k}}$ denotes an arbitrary graph with partial degree sequence 
$\vec{k}$, and $\Pcan\big(G^{\vec{k}}\big)$ is the canonical probability rewritten for one such graph. Indeed, \eqref{eq:PC} shows that the canonical probability is constant for all graphs with the same constraint, in our case, for all graphs with the same partial degree sequence. Using \eqref{eq:PC} and \eqref{partition} we can rewrite the canonical probability in the form
\be 
\Pcan\big(G^{\vec{k}}\big) 
= 2^{-{{n-m}\choose 2}}\prod_{i=1}^m \frac{{x_i^*}^{k_i}}{(1+x_i^*)^{n-m}} \prod_{1\le i<j \le m}(1+x_i^* x_j^*)^{-1},
\label{px}
\ee
where $x_i^* = e^{-\theta_i^*}$, and $\vec{\theta}^*=(\theta_i^*)_{i=1}^m$ is the 
vector of Lagrange multipliers coming from \eqref{canonical1}. 
Note that~\eqref{interp} can be rewritten as
\begin{equation}
S_n(\Pmic \mid \Pcan) = S\big(\,\delta[\vec{k^*}] \mid Q^n[\vec{k^*}]\,\big),
\label{SQ}
\end{equation} 
where $\delta[\vec{k^*}] = \prod_{i=1}^m \delta[k^*_i]$ is the \emph{multivariate Dirac distribution} 
with average $\vec{k^*}$. This has the interesting interpretation that the relative entropy between 
the distributions $\Pmic$ and $\Pcan$ \emph{on the set of graphs} $\cG_n$ coincides with the relative 
entropy between $\delta[\vec{k^*}]$ and $Q^n[\vec{k^*}]$ \emph{on the set of degree sequences}.
To be explicit, using \eqref{eq:QOmega} and~\eqref{px}, we can rewrite $Q^n[\vec{k^*}](\vec{k})$ as 
\begin{equation}
Q^n[\vec{k^*}](\vec{k}) =\Omega_{\vec{k}}^n\ 2^{-{{n-m}\choose 2}}\prod_{i=1}^m \frac{{x_i^*}^{k_i}}{(1+x_i^*)^{n-m}}
\prod_{1\le i<j \le m}(1+x_i^* x_j^*)^{-1}.
\label{eq:PoissonBinomial}
\end{equation}

The above distribution is a multivariate version of the \emph{Poisson-Binomial 
distribution} \cite{W93}. In the univariate case, 
the Poisson-Binomial distribution describes the probability of a certain number of successes 
out of a total number of independent and (in general) nonidentical Bernoulli trials~\cite{W93}. 
In this case, the marginal probability that node $i$ has degree $k_i$ in the canonical ensemble, 
irrespectively of the degree of any other node, is a univariate Poisson-Binomial given by the $m-1$ independent Bernoulli trials with success probabilities $\{p_{ij}^*\}_{1\le j\ne i \le m}$ and the $n-m$ independent Bernoulli trials with the same success probability $p_{i}^*$, with a total of $n-1$ independent Bernoulli trials. 
The relation in \eqref{SQ} can therefore be restated as
\begin{equation}
S_n(\Pmic \mid \Pcan) = S\big(\,\delta[\vec{k^*}] 
\mid \mathrm{PoissonBinomial}[\vec{k^*}]\,\big), 
\end{equation}
where $\mathrm{PoissonBinomial}[\vec{k^*}]$ is the multivariate Poisson-Binomial distribution 
given by ~\eqref{eq:PoissonBinomial}, i.e.,
\be
Q^n[\vec{k^*}] = \mathrm{PoissonBinomial}[\vec{k^*}].
\ee
The relative entropy can therefore be seen as coming from a situation in which the microcanonical 
ensemble forces the degree sequence to be exactly $\vec{k^*}$, while the canonical ensemble 
forces the degree sequence to be Poisson-Binomial distributed with average $\vec{k^*}$.
\paragraph{Two different regimes for the Poisson-Binomial distribution.}
It is known that the univariate Poisson-Binomial distribution admits two asymptotic limits: (1) a 
Poisson limit (if and only if $\sum_{j\ne i}p_{ij}^*\to\lambda>0$ and $\sum_{j\ne i}
(p_{ij}^*)^2\to0$ as $n\to\infty$~\cite{W93}); (2) a Gaussian limit (if and only if $p_{ij}^*\to\lambda_j>0$ 
for all $j\ne i$ as $n\to\infty$, as follows from a central limit theorem type of argument). If all 
the Bernoulli trials are identical, i.e., if all the probabilities $\{p_{ij}^*\}_{j\ne i}$ are equal, then 
the univariate Poisson-Binomial distribution reduces to the ordinary Binomial distribution, which 
also exhibits the well-known Poisson and Gaussian limits. These results imply that also the 
general multivariate Poisson-Binomial distribution in \eqref{eq:PoissonBinomial} admits a limiting 
behavior that should be consistent with the Poisson and Gaussian limits discussed above for
its marginals. This is precisely what we argue below.
\paragraph{Gaussian constrained degrees in the $\delta$-tame regime.}
Comparing \eqref{eq:deltalimit2} and~\eqref{Result}, and using \eqref{eq:alphatilde}, we see that Theorem~\ref{MainTheorem} shows that if the constraint is on the partial degree sequence, then
\be
\label{rescaling}
S_n(\Pmic \mid \Pcan) \sim m \log n\sim\log\sqrt{\det(2\pi Q)}
\ee 
whenever the regime is $\delta$-tame and condition \eqref{Condmn} is satisfied. Equation \eqref{rescaling} can be reinterpreted as the statement
\be
\label{extra} 
S_n(\Pmic \mid \Pcan) \sim S\big(\,\delta[\vec{k^*}] \mid \mathrm{Normal}[\vec{k^*},Q]\,\big),
\ee
where $\mathrm{Normal}[\vec{k^*},Q]$ is the \emph{multivariate Normal distribution} with 
mean $\vec{k^*}$ and covariance matrix $Q$. In other words, in the $\delta$-tame regime
\be 
Q^n[\vec{k^*}] \sim \mathrm{Normal}[\vec{k^*},Q],
\ee
i.e., the multivariate Poisson-Binomial distribution~\eqref{eq:PoissonBinomial} is asymptotically 
a multivariate Gaussian distribution whose covariance matrix is in general not diagonal, i.e., the 
dependence between the degrees of the different nodes does \emph{not} vanish. Since in this regime all the degrees are growing, so do all the eigenvalues of $Q$, 
which ensures that Formula~{\rm \ref{conj}} holds with $\tau_{\alpha_\infty}=0$, as proven 
in Theorem~\ref{MainTheorem}. 

Note that the right-hand side of \eqref{extra}, being the relative entropy of a discrete distribution 
with respect to a continuous distribution, needs to be properly interpreted: the Dirac distribution
$\delta[\vec{k^*}]$ needs to be smoothed to a continuous distribution with support in a small
ball around $\vec{k^*}$. Since the degrees are large, this does not affect the asymptotics. 

\paragraph{Poisson-Binomial unconstrained degrees in the $\delta$-tame regime.}
It is interesting to study the distribution of the degrees of the unconstrained nodes in the canonical ensemble. The canonical probability of the $(m+1)$-th degree (the first unconstrained node) can be computed and the same steps can be used to compute the canonical probabilities of the other unconstrained nodes, which follow the same probability law. The canonical probability that the $(m+1)$-th node is equal to some value $x\in \left\lbrace 0,1,\dots, n-1 \right\rbrace$ can be written as:
\be
\label{canonicalm+1}
\begin{aligned}
&\sum_{\substack{G\in\cG_n \\ k_{m+1}(G)=x}} \Pcan(G)\\
&=2^{-{{n-m}\choose 2}}\sum_{\substack{G\in\cG_n \\ k_{m+1}(G)=x}} \prod_{1 \leq i<j \leq m}\left( p_{ij}^* \right)^{g_{ij}(G)} \left( 1- p_{ij}^* \right)^{1-g_{ij}(G)}\prod_{i=1}^m\left( p_{i}^* \right)^{s_{i}(G)} \left( 1- p_{i}^* \right)^{n-m-s_{i}(G)}\\
&=2^{-{{n-m}\choose 2}}\sum_{G\in A}\prod_{i=1}^m {(p_{i}^*)}^{g_{i m+1}(G)} \left( 1- p_{i}^* \right)^{1-g_{i m+1}(G)}\\
&=2^{-{{n-m}\choose 2}}2^{{{n-m-1}\choose 2}}\sum_{G\in A\cap B}\prod_{i=1}^m {p_{i}^*}^{g_{i m+1}(G)} \left( 1- p_{i}^* \right)^{1-g_{i m+1}(G)}\\
&=\sum_{G\in A\cap B}\left(\tfrac12\right)^{(n-m-1)}\prod_{i=1}^m {p_{i}^*}^{g_{i m+1}(G)} \left( 1- p_{i}^* \right)^{1-g_{i m+1}(G)}\\
&=P(Po-Bi[p_1^*,\dots, p_m^*,\tfrac12,\dots, \tfrac12]=x),
\end{aligned}
\ee
where 
\be
\begin{aligned}
&A=\left\lbrace G\in\cG_n: k_{m+1}(G)=x,\ g_{ij}(G)=0\ \forall i=1,\dots,m,\ j=1,\dots,m, m+2,\dots,n,\ i\neq j \right\rbrace,\\
&B=\left\lbrace G\in\cG_n: g_{ij}(G)=0\ \forall i=m+2,\dots,n,\ j=m+2,\dots,n,\ i\neq j \right\rbrace,
\end{aligned}
\ee
and $Po-Bi[p_1^*,\dots, p_m^*,\tfrac12,\dots, \tfrac12]$ is the Poisson-Binomial distribution given by the $m$ independent trials $p_i^*$, $i=1,\dots,m$, and the $n-m-1$ independent Bernoulli trials with the same success probability $\tfrac12$. This means that, for each $j=m+1,\dots,n$, the canonical probability of the degree of the $j$-th node is distributed as a Poisson-Binomial random variable with $n-1$ entries: $p_1^*,\dots,p_1^*,\tfrac12,\dots, \tfrac12$.\\


\section{Proof of the main theorem}
\label{S2}
The proof is based on two lemmas, which are stated and proved in Section~\ref{S2.1}. In Section~\ref{S2.2} Theorem~\ref{MainTheorem} is proved.


\subsection{Preparatory lemmas}
\label{S2.1}

The following lemma gives an expression for the relative entropy. 

\begin{lemma}
\label{lemma2}
If the constraint is on the partial degree sequence $(k_i^*)_{i=1}^m$, then the relative entropy in \eqref{eq:KL2} equals
\begin{equation}
\label{RelEntro}
S_n(\Pmic \mid \Pcan) = \tfrac12\log[\det(2\pi Q)] - \log T^*,
\end{equation} 
where $Q$ is the covariance matrix in \eqref{covQ} and $T^*$ is the error in \eqref{eq:T}. The matrix $Q=(q_{ij})$ takes the form 
\begin{equation}
\label{Q}
\begin{cases}
q_{ii} = \sum_{1\le j\le m, j \neq i} p_{ij}^*(1-p_{ij}^*) + (n-m)p_i^*(1-p_i^*), \quad 1 \leq i \leq m,\\
q_{ij} = p_{ij}^*(1-p_{ij}^*), \quad 1 \leq i \neq j \leq m.
\end{cases}
\end{equation}
\end{lemma}

\begin{proof}
To compute $q_{ij}=\mathrm{Cov}_{\Pcan}(k_i,k_j)$, take the second order derivatives of 
the log-likelihood function
\be 
\begin{aligned}
&{\cal{L}}(\vec{\theta}) = \log\Pcan(G^* \mid \vec{\theta})\\
&= \log\left[2^{-{{n-m}\choose 2}}\prod_{1 \leq i < j \leq n} p_{ij}^{g_{ij}(G^*)} (1-p_{ij})^{(1-g_{ij}(G^*))} \prod_{i=1}^m p_{i}^{s_{i}(G^*)} \left( 1- p_{i} \right)^{n-m-s_{i}(G^*)}\right],
\end{aligned}
\ee
with
\be 
\quad p_{ij} = \frac{e^{-\theta_i - \theta_j}}{1+e^{-\theta_i - \theta_j}},\quad p_{i} = \frac{e^{-\theta_i}}{1+e^{-\theta_i}},
\ee
in the point $\vec{\theta}=\vec{\theta}^*$~\cite{GS}. 
It is easy to show that the first-order derivatives are~\cite{GL08}
\be
\frac{\partial}{\partial \theta_i}{\cal{L}}(\vec{\theta}\,)
= \langle k_i\rangle - k_i^*,
\quad 
\frac{\partial}{\partial \theta_i}{\cal{L}}(\vec{\theta}\,)\bigg|_{\vec{\theta}=\vec{\theta^*}}
= k_i^*-k_i^*=0
\ee
and the second-order derivatives are
\begin{equation}
\label{Covcalc}
\frac{\partial^2}{\partial \theta_i\partial \theta_j}{\cal{L}}(\vec{\theta})
\bigg|_{\vec{\theta}=\vec{\theta^*}} =\langle k_i\rangle\langle k_j\rangle - \langle k_i\,k_j\rangle 
= -\mathrm{Cov}_{\Pcan}(k_i,k_j).
\end{equation}
Taking the second-order derivatives of the log-likelihood function, we get \eqref{Q}.
The proof of \eqref{RelEntro} uses \cite[Formula 25]{GS}.
\end{proof}
The following lemma shows that a diagonal approximation of the matrix $Q$ is good for a $\delta$-tame partial degree sequence and $\alpha_n=m\log n$.

\begin{lemma}
\label{lemma3}
Under the $\delta$-tame condition,
\begin{equation}
\label{QQDcomp}
\log(\det Q_D) + o(m\log n) \leq \log(\det Q) \leq \log(\det Q_D)
\end{equation}
with $Q_D=\mathrm{diag}(Q)$ the matrix that coincides with $Q$ on the diagonal 
and is zero off the diagonal.
\end{lemma}

\begin{proof}
Use \cite[Theorem 2.3]{IL03}, which says that if
\begin{itemize}
\item[(1)] $\det(Q)$ is real, 
\item[(2)] $Q_D$ is non-singular with $\det(Q_D)$ real,
\item[(3)] $\lambda_i (A)>-1$, $1 \leq i \leq m$,
\end{itemize}
then 
\begin{equation}
\label{ILbds}
e^{-\frac{m\rho^2(A)}{1+\lambda_{\min}(A)}} \det Q_D \leq \det Q \leq \det Q_D.
\end{equation}
Here, $A=Q_D^{-1}Q_{\mathrm{off}}$, with $Q_{\mathrm{off}}$ the matrix that coincides with 
$Q$ off the diagonal and is zero on the diagonal, $\lambda_i(A)$ is the $i$-th eigenvalue of $A$
(arranged in decreasing order), $\lambda_{\mathrm{min}}(A) = \min _{1 \leq i \leq n}\lambda_i(A)$, 
and $\rho(A) = \max_{1 \leq i \leq n}|\lambda_i(A)|$.
\\
We verify (1)--(3).\\ 
\medskip\noindent
(1) Since $Q$ is a symmetric matrix with real entries, $\det Q$ exists and is real.

\medskip\noindent
(2) This property holds thanks to the $\delta$-tame condition and Lemma \ref{deltaTamep}. In fact 
\begin{equation}
\label{qij}
0 < \delta^{2} \leq q_{ij} \leq (1-\delta)^{2} < 1,
\end{equation}
and
\begin{equation}
\label{qii}
(m-1)\delta^2 + (n-m)\delta '^2 \leq q_{ii} = \leq (m-1)(1-\delta)^2 + (n-m)(1-\delta ')^2.  
\end{equation}

\medskip\noindent
(3) It is easy to show that $A=(a_{ij})$ is given by
\begin{equation}
\label{aijrel}
a_{ij} = \left\{\begin{array}{ll}
\frac{q_{ij}}{q_{ii}}=\frac{p_{ij}^*(1-p_{ij}^*)}{\sum_{1\le k\le m, k \neq i} p_{ik}^*(1-p_{ik}^*) + (n-m)p_i^*(1-p_i^*)}, &1 \leq i \neq j \leq m\\
0 &1 \leq i=j \leq m,
\end{array}
\right.
\end{equation}
where $q_{ij}$ is given by \eqref{Q}. The Gershgorin circle theorem says the eigenvalues of the matrix $A$ satisfy
\begin{equation}
\label{Gershgorin}
\mid\lambda_i(A)\mid \le R_i=\sum_{j\neq i}a_{ij} = \frac{\sum_{1\le k\le m, k \neq i} p_{ik}^*(1-p_{ik}^*)}{\sum_{1\le k\le m, k \neq i} p_{ik}^*(1-p_{ik}^*) + (n-m)p_i^*(1-p_i^*)},\quad 1\le i \le m.
\end{equation}
Using the $\delta$-tame condition, we find the bound
\begin{equation}
\label{MaxAV}
\mid\lambda_i(A)\mid \le\max_{1\le i \le m} R_i < 1-A(\delta),  
\end{equation}
with $A(\delta) = \frac{(n-m)\delta'}{(m-1)(1-\delta)^2 + (n-m)(1-\delta')}$. In principle, $A(\delta)$ also depends on $\delta'$, but $\delta'$ is itself function of $\delta$. Equation \eqref{MaxAV} immediately gives $\rho(A) < 1$, namely
\be
-\frac{m\rho^2(A)}{1+\lambda_{\min}(A)} > -\frac{m}{1+\lambda_{\min}(A)} .
\ee 
Next we show that
\be
-\frac{m}{1+\lambda_{\min}(A)} = o(m\log n).
\ee 
Together with \eqref{ILbds} this will settle the claim in \eqref{QQDcomp}. We must show that
\be
\lim_{n\to \infty} (1+\lambda_{\min}(A))\log n=\infty.
\ee
Using equation \eqref{MaxAV} again, it follows $1+\lambda_{\min}(A) > A(\delta)$. Therefore it suffices to prove that
\begin{equation}
\label{ToProve}
\lim_{n\to \infty}A(\delta)\log n = \infty.
\end{equation}
The result is trivial when $A(\delta)$ is constant ($\tfrac{n-m}{m}\to$ constant) or $A(\delta)\to\infty$ ($\tfrac{n-m}{m}\to \infty$) . On the other hand, when $A(\delta)\to 0$ ($\tfrac{n-m}{m}\to 0$), the condition $\tfrac{n-m}{m}\log n\to \infty$ is needed to conclude the proof.
\end{proof}


\subsection{Proof of Theorem~\ref{MainTheorem}}
\label{S2.2}

\begin{proof}
When $\alpha_n = \tfrac{m\log n}{2}$, Lemma~\ref{lemma2} says 
\begin{equation}
\label{finalrelent}
\lim_{n \to \infty} \frac{S_n(\Pmic \mid \Pcan)}{\alpha_n} 
= \lim_{n \to \infty}\frac{\log 2\pi}{m\log n} 
+ \lim_{n \to \infty}\frac{\log(\det Q)}{m\log n} - \lim_{n \to \infty}\frac{\log T^*}{2m\log n}.
\end{equation}
The last term (the error) tends to zero. In fact, in \cite{GS} it is proved that $\lim_{n \to \infty}\frac{\log T^*}{m\log n}=0$ unless the number of eigenvalues of $Q$ that have a finite limit as $n \to \infty$ which is indeed the case when a partial $\delta$-tame degree sequence is put as a constraint and $\alpha_n = \tfrac{m\log n}{2}$.\\
Using the $\delta$-tame condition, we get from Lemma~\ref{lemma3} that
\begin{equation}
\lim_{n \to \infty} \frac{\log(\det Q)}{m\log n} 
= \lim_{n \to \infty}\frac{\log(\det Q_D)}{m\log n}.
\end{equation}
To conclude the proof it therefore suffices to show that
\begin{equation}
\label{qdnf}
\lim_{n \to \infty} \frac{\log(\det Q_D)}{m\log n} = 1.
\end{equation} 
Using \eqref{qii}, we have
\begin{equation}
\label{qdnf2}
\begin{aligned}
&\frac{\log[(m-1)\delta^2 + (n-m)\delta']}{\log n} \leq \frac{\sum_{i=1}^m\log q_{ii}}{m\log n} 
= \frac{\log(\det Q_D)}{m\log n}\\ 
&\leq \frac{\log[(m-1)(1-\delta)^2 + (n-m)(1-\delta')]}{\log n}.
\end{aligned}
\end{equation}
Both sides tend to 1 as $n\to\infty$, and so \eqref{qdnf} follows. 
\end{proof}


\appendix


\section{Appendix}
\label{appA}
In this appendix we identify the structure of the canonical ensemble when the constraint is put on the partial degree sequence for the first $m<n$ nodes. The partial degree sequence $\vec{k}(G) = (k_i(G))_{i=1}^m$ is set to a \emph{specific $m$-dimensional} vector $\vec{k}^*$, which is assumed to be \emph{graphical}, i.e., there is at least one graph $G^*\in\cG_n$ with partial degree sequence $\vec{k}^*$. The constraint is therefore
\be
\vC^* = \vec{k}^*= (k_i^*)_{i=1}^m \in \{1,2,\dots ,n-2\}^m.
\ee 
The canonical ensemble has Hamiltonian $H(G,\vec{\theta})
=\sum_{i=1}^m \theta_i k_i(G)$, where $G$ is a graph belonging to $\cG_{n}$, and 
$k_i(G) = \sum_{j\neq i} g_{ij}(G)$ is the degree of node $i$. 
It is easy to transform the Hamiltonian into 
\begin{equation}
H(G,\vec{\theta}) = \sum_{1\le i<j \le m}(\theta_i + \theta_j)g_{ij}(G) + \sum_{i=1}^m \theta_i \sum_{j=m+1}^n g_{ij}(G) 
\end{equation}
Using this form, we see that the partition function equals 
\be
\label{partition}
\begin{aligned}
&Z(\vt)=\sum_{G\in\cG_{n}} e^{-H(G,\vec{\theta})}= \sum_{G\in\cG_{n}}\prod_{1\le i<j\le m} e^{-(\theta_i+\theta_j) g_{ij}(G)}\prod_{i=1}^m\prod_{j=m+1}^n e^{-\theta_i g_{ij}(G)} \\
&=2^{{n-m}\choose 2}\prod_{1\le i<j\le m} (1+e^{-(\theta_i+\theta_j)})\prod_{i=1}^m\prod_{j=m+1}^n (1+e^{-\theta_i})
\\
&=2^{{n-m}\choose 2}\prod_{1\le i<j\le m} (1+e^{-(\theta_i+\theta_j)})\prod_{i=1}^m (1+e^{-\theta_i})^{(n-m)}.
\end{aligned}
\ee
Inserting the partition function into the canonical expression, we get
\begin{equation}
\Pcan(G\mid \theta) = 2^{-{{n-m}\choose 2}}\prod_{1 \leq i<j \leq m} p_{ij} ^{g_{ij}(G)} \left( 1- p_{ij} \right)^{1-g_{ij}(G)}\prod_{i=1}^m p_{i} ^{s_{i}(G)} \left( 1- p_{i} \right)^{n-m-s_{i}(G)}
\end{equation}
with 
\be
p_{ij} = \frac{e^{-\theta_i-\theta_j}}{1 + e^{-\theta_i-\theta_j}},\qquad p_{i} = \frac{e^{-\theta_i}}{1 + e^{-\theta_i}},\qquad s_i(G)=\sum_{j=m+1}^n g_{ij}(G).
\ee 
It remains to tune the Lagrange multipliers to the values such that the average constraint equals the vector $\vC^* = \vec{k}^*= (k_i^*)_{i=1}^m \in \{1,2,\dots ,n-2\}^m$.
The average energy of the $i$-th degree with respect to the probability distribution $\Pcan(\cdot \mid \theta)$ corresponds to the derivative with respect to $\theta_i$ of the logarithm of the partition function (free energy). This means that the values $(\theta^*_i)_{i=1}^m$ must satisfy
\begin{equation}
\langle k_i \rangle = \sum_{1\le j\le m\atop j \neq i}p_{ij}^* + (n-m)p_{i}^* = k_i^*, \qquad 1\le i\le m
\end{equation}
with 
\be
p_{ij}^* = \frac{e^{-\theta_i^*-\theta_j^*}}{1 + e^{-\theta_i^*-\theta_j^*}},\qquad p_{i}^* = \frac{e^{-\theta_i^*}}{1 + e^{-\theta_i^*}}.
\ee 
The canonical ensemble therefore takes the form
\begin{equation}
\label{CanonicalMixed}
\Pcan(G) = 2^{-{{n-m}\choose 2}}\prod_{1 \leq i<j \leq m} {p^*_{ij}} ^{g_{ij}(G)} \left( 1- {p^*_{ij}} \right)^{1-g_{ij}(G)}\prod_{i=1}^m {p^*_{i}}^{s_{i}(G)} \left( 1- {p^*_{i}} \right)^{n-m-s_{i}(G)}.
\end{equation}
The expression in \eqref{CanonicalMixed} has an interpretation. Indeed, the canonical formula che be split into two parts: 
\be
\Pcan(G)=\Pcan^U(G) \Pcan^B(G)\ 2^{-{{n-m}\choose 2}},
\ee
with 
\be 
\Pcan^U(G)=\prod_{1 \leq i<j \leq m} {p^*_{ij}} ^{g_{ij}(G)} \left( 1- {p^*_{ij}} \right)^{1-g_{ij}(G)}
\ee
and 
\be 
\Pcan^B(G)=\prod_{i=1}^m {p^*_{i}}^{s_{i}(G)} \left( 1- {p^*_{i}} \right)^{n-m-s_{i}(G)}.
\ee
The unipartite probability, $\Pcan^U(G)$, is the canonical probability obtained when the constraint is put on the full degree sequence $\vec{u}^*= (u_i^*)_{i=1}^m$ on the set $\cG_m$. The constrained degree sequence is precisely $u_i^* = \sum_{1\le j\le m\atop j \neq i}p_{ij}^*$. The bipartite probability $\Pcan^B(G)$ is the canonical bipartite probability obtained when the constraint is put only on the top layer of a bipartite graph. In this case the configuration space is the set of bipartite graphs $\cG_{m,n-m}$ with $m$ nodes on the top layer and $n-m$ nodes on the bottom layer. The constrained top layer degree sequence is $\vec{b}^*= (b_i^*)_{i=1}^m$, where $b_i^* = (n-m)p_i^*$. The third factor $2^{-{{n-m}\choose 2}}$ is the inverse of the number of possible (unconstrained) graphs with $n-m$ nodes. In conclusion, the canonical probability in \eqref{CanonicalMixed} can be interpreted as the product of two canonical probabilities, $\Pcan^U(G)$ and $\Pcan^B(G)$, and the number $2^{-{{n-m}\choose 2}}$. Both canonical probabilities have an $m$-dimensional degree sequence as a constraint $\vec{u}^*= (u_i^*)_{i=1}^m$ and $\vec{b}^*= (b_i^*)_{i=1}^m$, put on the respective configuration spaces. Furthermore, two degree sequences sum up to the original degree sequence, namely,
\be 
u_i^* + b_i^* = k_i^* \qquad \forall i=1,\dots,m.
\ee
For this reason $(p_{ij}^*)_{i,j=1}^m$ are called the unipartite probabilities and $(p_{i}^*)_{i=1}^m$ the bipartite probabilities.


\section{Appendix}
\label{appB}
In this appendix we identify the structure of the $\delta$-tame condition when a partial degree sequence $(k_i^*)_{i=1}^m$ is put as a constraint on $\cG_n$. The definition comes from the situation where a full degree sequence $(k_i^*)_{i=1}^n$ is fixed on $\cG_n$ \cite{BH}. In the full degree sequence situation the canonical probability takes the form 
\begin{equation}
\label{canonicalO}
\Pcan(G) = \prod_{1 \leq i<j \leq n}\left( p_{ij}^* \right)^{g_{ij}(G)} \left( 1- p_{ij}^* \right)^{1-g_{ij}(G)}
\end{equation}
with 
\be
\label{canonical1O} 
p_{ij}^* = \frac{e^{-\theta_i^*-\theta_j^*}}{1 + e^{-\theta_i^*-\theta_j^*}}\quad {\forall\ i\neq j},
\ee 
and with the vector of Lagrange multipliers $\vec{\theta}^*=(\theta_i^*)_{i=1}^n$ tuned such that 
\begin{equation}
\label{canonical2O}
\langle k_i \rangle = \sum_{1\le j\le n\atop j \neq i}p_{ij}^* = k_i^*, \qquad 1\le i\le n.
\end{equation}
The degree sequence $(k_i^*)_{i=1}^n$ is said to be $\delta$-tame when there exists a $\delta\in (0,\tfrac12]$ such that, for each $1\le i\neq j\le n$, the canonical probabilites satisfy 
\be 
\delta < p_{ij}^* < 1-\delta.
\ee
\begin{definition}[$\delta$-tame partial degree sequence]
\label{deltaTamep}
We say that such a sequence is $\delta$-tame when there exists a $\delta\in (0,\tfrac12]$ such that, for each $1\le i\neq j\le m$, the canonical probabilites defined in \eqref{canonical}--\eqref{canonical2} satisfy 
\be
\delta < p_{ij}^* < 1-\delta \qquad \forall\ 1\le i\neq j \le m.
\ee
\end{definition}
\begin{lemma}
If $(k_i^*)_{i=1}^m$ is a partial degree sequence on $\cG_n$ and it is $\delta$-tame in the sense of Definition \ref{deltaTamep}, then the canonical bipartite probabilities satisfy 
\be
\delta' < p_{i}^* < 1-\delta' \qquad \forall\ 1\le i\neq j \le m,
\ee
for some $\delta'\in (0,\tfrac12]$.
\end{lemma}
\begin{proof}
The canonical probabilities, tuned with the proper $(\theta_{i}^*)_{i=1}^m$, satisfy
\be
p_{ij}^*=\frac{x_i x_j}{1+x_i x_j},\qquad p_{i}^* = \frac{x_i}{1 + x_i},\quad x_i=e^{-\theta_i^*}.
\ee
Since $(k_{i})_{i=1}^m$ is a partial $\delta$-tame degree sequence, Definition \ref{deltaTamep} says that
\be 
\delta < p_{ij}^* <(1-\delta).
\ee
From this it follows that
\be 
\label{deltaproof1}
\frac{\delta}{1-\delta}<x_i x_j< \frac{1-\delta}{\delta}.
\ee
Using \eqref{deltaproof1} for different indices $i,j,k$, we get
\be 
\left(\frac{\delta}{1-\delta}\right)^2<x_i^2 x_j x_k<\left( \frac{1-\delta}{\delta}\right)^2.
\ee
Using again \eqref{deltaproof1} for the indices $j$ and $k$, we get 
\be 
\label{deltaproof2}
\left(\frac{\delta}{1-\delta}\right)^{3/2}<x_i<\left( \frac{1-\delta}{\delta}\right)^{3/2}.
\ee
Using \eqref{deltaproof2} and $p_{i}^* = \frac{x_i}{1 + x_i} = \frac{1}{1 + \frac{1}{x_i}}$, we obtain that 
\be 
\delta' < p_i^* < 1-\delta'
\ee
with $\delta' = \frac{1}{1+(\frac{1-\delta}{\delta})^{3/2}}$. Note that $0<\delta\le \tfrac12$ implies $0<\delta'\le \tfrac12$.
\end{proof}


\end{document}